\documentclass[12pt]{article}

\usepackage{color,url}
\usepackage{graphicx}
\usepackage{geometry} 
\usepackage{amsmath,amssymb,amsfonts,bbm,amsthm}

\newcommand{\Es}{\mbox{\bf E}}    

\newtheorem{thm}{Theorem}
\newtheorem{remark}{Remark}

\title{Parameter estimation for the discretely observed fractional Ornstein-Uhlenbeck process and the Yuima R package}
\author{Alexandre Brouste\footnote{E-mail: \texttt{alexandre.brouste@univ-lemans.fr}, corresponding author.}\\ Laboratoire Manceau  de Math\'ematiques\\Universit\'e du Maine\\Avenue Olivier Messiaen - 72100 Le Mans, France \and Stefano M. Iacus\footnote{E-mail: \texttt{stefano.iacus@unimi.it}}\\Department of Economics, Business and Statistics\\University of Milan\\Via Conservatorio, 7 - 20122 Milan, Italy}
\date{}

\begin{document}

\maketitle

\begin{abstract}
This paper proposes consistent and asymptotically Gaussian estimators for the parameters $\lambda$, $\sigma$ and $H$ of the discretely observed fractional Ornstein-Uhlenbeck process solution of the stochastic differential equation $d Y_t = -\lambda Y_t dt + \sigma d W_t^H$, where $(W_t^H, t\geq 0)$ is the fractional Brownian motion. For the estimation of the drift $\lambda$, the results are obtained only in the case when $\frac12 < H < \frac34$. This paper also provides ready-to-use software for the R statistical environment based on the YUIMA package.
\end{abstract}

\section{Introduction}

Statistical inference for parameters of ergodic diffusion processes observed on discrete increasing
grid have been much studied. Local asymptotic normality (LAN) property of the likelihoods
have been shown in \cite{gobet1} for elliptic ergodic diffusion, under proper conditions for
the drift and the diffusion coefficient, and a mesh satisfying
$$ \Delta_N \longrightarrow 0 \quad \mbox{and} \quad N\Delta_N \longrightarrow +\infty $$
when the size of the sample $N$ grows to infinity. 
Estimation procedure have been studied by many authors, mainly in the one-dimensional case
(see, for instance, \cite{FZ, Kess} and \cite{Yoshida1} in the multidimensional setting). All estimators in the previous works
are based on contrasts (for contrasts framework, see \cite{GC}), assuming in the general case, that for 
some $p >1$, as $n \longrightarrow +\infty$, $N \Delta_N^p \longrightarrow 0$. In particular, for Ornstein-Uhlenbeck process,  transitions densities are known,
and all have been treated, as remarked in \cite{Jacod1}. 

In the fractional case, we consider the fraction Ornstein-Uhlenbeck process (fOU),  the solution of
$$dY_t = -\lambda Y_t dt + \sigma dW_t^H $$
where $W^H=\left( W^H_t, t\geq 0 \right)$ is a normalized fractional Brownian motion (fBM), {\it i.e.} the zero mean Gaussian processes with covariance function
$$\Es W_s^H W_t^H = \frac{1}{2} \left( |s|^{2H}+|t|^{2H}- |t-s|^{2H}\right)$$ with Hurst exponent $H\in(0,1)$.  The fOU process is  neither Markovian nor a semimartingale for $H\neq\frac{1}{2}$ but
remains Gaussian and ergodic (see \cite{cheridito}).  For $H > \frac12$, it even presents the long-range dependance property that makes it useful for different applications in biology, physics, ethernet traffic or finance.

Statistical large sample properties of Maximum Likelihood Estimator of the drift parameter in the continuous observations case have been treated in \cite{Coutin2,BK10, Lototsky2,Klep02a} for different applications. Moreover, asymptotical properties of the Least Squares Estimator have been studied in~\cite{hunual}. 

In the discrete case and fractional case, we can cite few works on the topic. On the one hand, very recent works give methods to estimate the drift $\lambda$  by contrast procedure \cite{ludena,Neue} or the drift $\lambda$ and the diffusion coefficient $\sigma$ with discretization procedure of integral transform \cite{chinois}. In these papers, the Hurst exponent is supposed to be known and only consistency is obtained.   On the other hand, methods to estimate the Hurst exponent $H$ and the diffusion coefficient are presented in \cite{berzinleon} with classical order 2 variations convolution filters.  

To the best of our knowledge, nothing have been done,  to have a complete estimation procedure
that could estimate all Hurst exponent, diffusion coefficient and drift parameter with central limit theorems and this is the gap we fill in this paper. Moreover, estimates of $H$, $\sigma$ and $\lambda$ presented in this paper slightly differ from all those studied previously.

In Section \ref{sec1} we review the basic facts of stochastic differential equations driven by the fractional Brownian motion and we introduce the basic notations and assumptions.
Section \ref{sec2} presents consistent and asymptotically Gaussian estimators of the parameters of the fractional Ornstein-Uhlenbeck process from discrete observations. In Section \ref{sec3} we present ready-to-use software  for the R statistical environment which allows the user to simulate and estimate the parameters of the fOU process. We further present Monte-Carlo experiments to test the performance of the estimators under different sampling conditions.

\section{Model specification}\label{sec1}

Let $X=(Y_t, t\geq 0)$ be a fractional Ornstein-Uhlenbeck process (fOU), {\it i.e.} the solution of 
\begin{equation}\label{eq:model}
Y_t = y_0 -\lambda \int_0^t Y_s ds + \sigma W^{H}_t,  \quad t >0, \quad Y_0=y_0,
\end{equation}
where unknown parameter $\vartheta=\left(\lambda,\sigma,H\right)$ belongs to an open subset $\Theta$ of $(0,\Lambda)\times [ \underline{\sigma},\overline{\sigma}] \times (0,1) $, $0<\Lambda <+\infty$, $0< \underline{\sigma}<\overline{\sigma}< +\infty$ and $W^H=(W^H_t, t\geq 0)$ is a standard fractional Brownian motion \cite{kolmogorov,mandel}
of Hurst parameter $H\in(0,1)$, {\it i.e.} a Gaussian centered process of covariance function
$$ \Es W^H_t W^H_s=\frac{1}{2}\left( t^{2H} + s^{2H} - |t-s|^{2H}  \right)  .$$
It is worth emphasizing that in the case $H=\frac{1}{2}$, $W^\frac{1}{2}$ is the classical Wiener process
The fOU process is  neither Markovian nor a semimartingale for $H\neq\frac{1}{2}$ but
remains Gaussian and ergodic.  For $H > \frac12$, it even presents the long-range dependance property  (see \cite{cheridito}).

The present work exposes an estimation procedure for estimating all three components of
$\vartheta$ given the regular discretization of the sample path $Y^T=(Y_t, 0\leq t\leq T)$, precisely
$$\left( X_n : = Y_{n\Delta_N} , n=0\ldots N\right)\,,$$
where $T=T_N = N \Delta_N \longrightarrow +\infty$ and $\Delta_N \longrightarrow 0$ as $N\longrightarrow +\infty$. 

In the following, convergences $\stackrel{\cal L}{\longrightarrow}$, $\stackrel{p}{\longrightarrow} $ and $\stackrel{a.s.}{\longrightarrow} $ stand  respectively for the convergence in law, the convergence in probability
and the  almost-sure convergence.

\section{Estimation procedure}\label{sec2}
Contrary to the previous works on the subject, we consider here the problem of estimation of $H$, $\sigma$ and $\lambda$ when all parameters are unknown,  using discrete observations from the fractional Ornstein-Uhlenbeck process. Due to the fact that one can estimate  $H$ and $\sigma$ without the knowledge of $\lambda$, our approach consists naturally in a two step procedure.
\subsection{Estimation of the Hurst exponent $H$ and the diffusion coefficient $\sigma$ with quadratic generalized variations}\label{sec:H}

The key point of this paper is that the Hurst exponent $H$  and the diffusion coefficient $\sigma$ can be estimated 
without estimating $\lambda$.

Let $\mathbf{a}=(a_0,\ldots,a_K)$ be a discrete filter of length
$K+1$, $K\in\mathbb{N}$, and of order $L\geq 1$, $K\geq L$, {\it i.e.}
\begin{equation}
\sum_{k=0}^K a_k k^\ell=0\quad \mbox{for} \quad0\le \ell\le 
L-1 \,\,\,\,\,\,\mbox{and}\,\,\,\,\,\,   \sum_{k=0}^K a_k k^{L}\neq 0.
\end{equation}
Let it be normalized with
\begin{equation}
\sum_{k=0}^{K} (-1)^{1-k} a_k = 1\,.
\end{equation}
In the following, we will also consider dilatated filter $\mathbf{a}^2$ associated to $\mathbf{a}$ defined by
$$a^2_k=\left\{\begin{array}{cc} a_{k'}&\mbox{ if } k=2k'\\
0&\mbox{ sinon. }\end{array}\right. \quad \mbox{ for } \quad 0\le k\le 2K \,.$$
Since $\underset{k=0}{\overset{2K}{\sum}}a_k^2 k^r=2^r\underset{k=0}{\overset{K}{\sum}}
k^ra_k$, filter $\mathbf{a}^2$ as the same order than  $\mathbf{a}$.  We denote by
$$ V_{N,\mathbf{a}} = \sum_{i=0}^{N-K} \left( \sum_{k=0}^K  a_k X_{i+k} \right)^2$$
the generalized quadratic variations associated to the filter $\mathbf{a}$ (see for instance \cite{istaslang}) and,  finally, 
$$ \widehat{H}_N = \frac{1}{2}\log_2 \frac{V_{N,\mathbf{a}^2}}{V_{N,\mathbf{a}}} $$
and
$$ \widehat{\sigma}_N = \left( 2 \cdot - \frac{V_{N,\mathbf{a}}}{\sum_{k,\ell} a_k a_\ell |k-\ell |^{2\widehat{H}_N} \Delta_N^{2\widehat{H}_N} }\right)^{\frac{1}{2}} .$$

\vskip 24pt
\begin{thm}\label{thm:Hs} Let $\mathbf{a}$ be a filter of order $L \geq 2$.  Then, both estimators $\widehat{ H}_N$ and  $\widehat{\sigma}_N$ are  strongly consistent, {\it i.e.}
$$ (\widehat{ H}_N ,  \widehat{\sigma}_N) \stackrel{a.s.}{\longrightarrow}   (H, \sigma) \quad \mbox{as $N \longrightarrow +\infty$}.$$
Moreover, we have asymptotical normality property, {\it i.e.} as $N \rightarrow +\infty$, for all $H \in (0,1)$,
$$  
\sqrt{N}( \widehat{H}_N-H) \stackrel{\cal L}{\longrightarrow}  {\cal N} (0, \Gamma_1(\vartheta,\mathbf{a}))
$$
and
$$
 \frac{\sqrt{N}}{\log N} ( \widehat{\sigma}_N-\sigma) \stackrel{\cal L}{\longrightarrow}  {\cal N} (0, \Gamma_2(\vartheta,\mathbf{a}))
$$
where  $\Gamma_1(\vartheta,\mathbf{a})$ and $\Gamma_2(\vartheta,\mathbf{a})$  symmetric definite positive matrices depending on $\sigma$, $H$, $\lambda$ and the filter $\mathbf{a}$. 
\end{thm}
\begin{proof}

The solution of \eqref{eq:model} can be explicited
$$ Y_t = x_0 e^{-\lambda t} + \sigma \int_0^t e^{-\lambda(t-s)} dW^{H}_s\,, $$
where the integral is defined as a Riemann-Stieljes pathwise integral. Let us consider the stationary centered Gaussian solution
$$ Y^\dag_t = \sigma \int_{-\infty}^t e^{-\lambda(t-u)} dW^{H}_u.$$
We have also, 
$$ Y^\dag_t - Y_t = e^{-\lambda t} \left( Y^\dag_0 - y_0 \right) \stackrel{a.s.}{\longrightarrow}  0.$$
It is known (see \cite[Lemma 2.1]{cheridito}) that
\begin{eqnarray*}
\Es Y_0^\dag ( Y_0^\dag - Y_t^\dag)&=& - \sigma^2 H(2H-1) e^{-\lambda t} \int_{-\infty}^0 e^{\lambda u}  \left( \int_{0}^t e^{\lambda v} (v-u)^{2H-2} dv \right) du.
\end{eqnarray*}
Let $v(t)$ denote the variogram of $Y_t^\dag$. We now show that 
$$v(t) =  \Es \left(Y_0^{\dag}\right)^2 - \Es Y^\dag_t Y^\dag_0 =  \frac{ \sigma^2}{2} |t|^{2H} + r(t)$$
where $r(t)=o(|t|^{2H})$  as $t$ tends to zero. Indeed,
\begin{eqnarray*}
v(t)& =&  -\sigma^2 H(2H-1)   \int_{-\infty}^0 e^{\lambda u}  \left( \int_{0}^t e^{-\lambda (t- v)} (v-u)^{2H-2} dv \right) du \\
& =&  -\sigma^2 H(2H-1)   \int_{-\infty}^0 e^{\lambda u}  \left( \int_{0}^t e^{-\lambda r} (t-r-u)^{2H-2} dr \right) du \\
& =&  -\sigma^2 H(2H-1)   \int_{0}^{\infty}    \int_{0}^t e^{-\lambda (r+u)} (t-r+u)^{2H-2} dr  du \\ 
& =&  -\sigma^2 H(2H-1)   \int_{0}^{\infty}    \int_{u}^{u+t} e^{-\lambda w} (t-w+2u)^{2H-2} dw  du \\ 
& =&  -\sigma^2 H(2H-1)   \int_{0}^{\infty}   e^{-\lambda w} \left( \int_{\max(0,w-t)}^{w} (t-w+2u)^{2H-2} du \right) dw \\
& =&  -\frac{1}{2}\sigma^2 H(2H-1)   \int_{0}^{\infty}   e^{-\lambda w} \left( \int_{|t-w|}^{t+w} x^{2H-2} dx \right) dw. 
\end{eqnarray*}
Thus,
\begin{eqnarray*}
\frac{dv}{dt}(t)&=& -\frac{1}{2}\sigma^2 H  (2H-1) \int_{0}^{\infty}   e^{-\lambda w} \left(  (t+w)^{2H-2}-|t-w|^{2H-2}\right) dw \\
&=& -\frac{1}{2}\sigma^2 H  (2H-1)  t^{2H-1}  \int_{0}^{\infty}   e^{-\lambda ty } \left(  (1+y)^{2H-2}-|1-y|^{2H-2}\right) dy  \\
 &=&-\frac{1}{2}\sigma^2 H (2H-1)  t^{2H-1}  \underbrace{\int_{0}^{\infty} \left(  (1+y)^{2H-2}-|1-y|^{2H-2}\right) dy}_{<\infty} + \tilde{r}(t) \\
 & =& \sigma^2 H   t^{2H-1} +  \tilde{r}(t)\end{eqnarray*}
with 
$$ \tilde{r}(t) = -\frac{1}{2}\sigma^2 H (2H-1)  t^{2H-1} \sum_{i=1}^\infty \int_{0}^{\infty} \frac{(-\lambda t y)^i}{i !} \left(  (1+y)^{2H-2}-|1-y|^{2H-2}\right) dy.$$
Therefore, we proved that 
$$ v(t) =  \frac{ \sigma^2}{2} |t|^{2H} + r(t).$$

Now, applying results in \cite[Theorem 3(i)]{istaslang}, the proof of  Theorem \ref{thm:Hs} is complete because the following conditions are fulfilled:
 \begin{itemize}
 \item firstly, $r(t)=o(|t|^{2H})$  as $t$ tends to zero,
\item secondly, for classical generalized quadratic variations or order $L\geq 2$ (for instance $L=2$),
$$ |r^{(4)}(t)|\leq G |t|^{2H+1-\varepsilon-4}$$ 
with $2H+1 -\varepsilon > 2H$ and $4 > 2H+1 -\varepsilon+1/2$ for  $\varepsilon<1$
and any $H\in(1/2,1)$.
\end{itemize}
\end{proof}

\begin{remark}\label{rem:1} We have two useful examples of filters. Classical filters  of order $L\geq1$ are defined by 
$$a_k=c_{L,k}=\frac{(-1)^{1-k}}{2^K}\begin{pmatrix} K\\k
\end{pmatrix}=\frac{(-1)^{1-k}}{2^K}\frac{K!}{k!(K-k)!}
\quad  \mbox{ pour } \quad 0\le k\le K .$$

Daubechies filters of even order can also be considered (see \cite{daub}), for instance
the order 2 Daubechies' filter:
$$\frac{1}{\sqrt{2}}(.4829629131445341,-.8365163037378077,.2241438680420134,.1294095225512603).$$
\end{remark}

\begin{remark} For classical order 1 quadratic variations ($L=1$)  and $\mathbf{a}=\left( -\frac{1}{2},\frac{1}{2} \right)$ we can also obtain consistency for any value of $H$, but the central limit theorem holds only for $H < \frac34$ (see \cite{istaslang}).

\end{remark}




\subsection{Estimation of the drift parameter $\lambda$ when both $H$ and $\sigma$ are unknown}\label{sec:lambda}
From  \cite{hunual}, we know the following result
$$\lim_{t\longrightarrow \infty} \mbox{var}(Y_t)=\lim_{t\longrightarrow \infty} \frac{1}{t} \int_0^t Y_t^2 dt=  \frac{\sigma^2 \Gamma\left( 2H +1\right)}{2\lambda^{2H}}=: \mu_2\,.$$
This gives a natural plug-in estimator of $\lambda$, namely
$$\widehat{\lambda}_N=  \left( \frac{2 \,\widehat{\mu}_{2,N}}{\widehat{\sigma}_N^2 \Gamma\left(2 \widehat{H}_N+1\right)}\right)^{-\frac{1}{2\widehat{H}_N}} $$
where  $\widehat{\mu}_{2,N}$ is the empirical moment of order 2, {\it i.e} 
$$ \widehat{\mu}_{2,N} = \frac{1}{N} \sum_{n=1}^N X^2_n.$$

\begin{thm}\label{thm:lambda} Let $H\in\left( \frac{1}{2},\frac{3}{4}\right)$ and a mesh satisfying the condition 
$ N \Delta_N^p\longrightarrow 0$, $p > 1$, as $N \longrightarrow + \infty$. Then, as $N \longrightarrow + \infty$,
$$ \widehat{\lambda}_N  \stackrel{a.s.}{\longrightarrow}  \lambda $$
and $$ \sqrt{T_N} \left(  \widehat{\lambda}_N  - \lambda\right)\stackrel{\cal L}{\longrightarrow} {\cal N} (0, \Gamma_3(\vartheta)),
$$
where $\Gamma_3(\vartheta)= \lambda \left( \frac{\sigma_H}{2H} \right)^2 $  and
\begin{equation}\label{eq:sigma}
 \sigma^2_H=(4H-1) \left( 1+ \frac{\Gamma(1-4H) \Gamma(4H-1)}{\Gamma(2-2H)\Gamma(2H)}\right).
 \end{equation}

\end{thm}
\begin{proof} Let us note $T_N=N \Delta_N$.
It had been shown in \cite{hunual} that, as $T_N \rightarrow +\infty$ (or as $N  \rightarrow +\infty$),
\begin{equation}\label{eq:mu1} 
\frac{1}{T_N} \int_0^{T_N} X_t^2 dt  \stackrel{a.s.}{\longrightarrow}  \kappa_H \lambda^{-2H}
\end{equation}
and, with straightforward calculus,
\begin{equation}\label{eq:mu2} 
 \sqrt{T_N} \left( \frac{1}{T_N} \int_0^{T_N} X_t^2dt  - \kappa_H \lambda^{-2H} \right) \stackrel{\cal L}{\longrightarrow}  {\cal N}( 0, \left( \sigma_H \kappa_H \right)^2 \lambda^{-4H-1} ) 
 \end{equation}
where $ \kappa_H=\sigma^2\frac{\Gamma(2H+1)}{2}$  and $\sigma_H$ is defined by \eqref{eq:sigma}. 
Let us denote $ \widehat{\mu}_{2,N}$ the discretization of the integral 
$$  \widehat{\mu}_{2,N} = \frac{1}{N} \sum_{n=1}^N X^2_n \quad \mbox{and} \quad  
\mu_2= \kappa_H \lambda^{-2H}.$$
Then
$$
\sqrt{T_N} \left(   \widehat{\mu}_{2,N}  -\mu_2  \right)  =  \sqrt{T_N} \left(   \widehat{\mu}_{2,N}  - \frac{1}{T_N} \int_0^{T_N} X_t^2dt  \right) + \sqrt{T_N} \left(   \frac{1}{T_N} \int_0^{T_N} X_t^2dt   -\mu_2  \right).$$
As $\left( X_t, t \geq 0\right)$ is a Gaussian process and H\"{o}lder of order $\frac{1}{2} < H < \frac{3}{4}$, we have $\sqrt{T_N} \left(   \widehat{\mu}_{2,N}  - \frac{1}{T_N} \int_0^{T_N} X_t^2dt  \right) \stackrel{p}{\longrightarrow}  0$  as $N \longrightarrow +\infty$ provided that 
$N \Delta_N^p\longrightarrow 0$, $p >1$, (see \cite[Lemma 8]{Kess}), we deduce from  \eqref{eq:mu1} and  \eqref{eq:mu2} that
\begin{equation}
\label{eq:conv6}
\sqrt{T_N} \left(   \widehat{\mu}_{2,N}  -\mu_2  \right) \stackrel{\cal L}{\longrightarrow}  {\cal N}\left( 0, \left( \sigma_H \kappa_H \right)^2 \lambda^{-4H-1}\right).
\end{equation}
Let us introduce the following two quanitites
$$ M_N= \begin{pmatrix}
  \widehat{\mu}_{2,N}\\
\widehat{H}_N \\
\widehat{\sigma}_N \\
\end{pmatrix} \quad \mbox{and}  \quad m=\begin{pmatrix} \mu_2 \\ H \\  \sigma \end{pmatrix}  
.$$
Finally, results obtained in Theorem \ref{thm:Hs} and the convergence in \eqref{eq:conv6} gives consistency of $M_N$, i.e.  $M_N \stackrel{P}{\longrightarrow} m$ as as $N \longrightarrow + \infty$.
Let us further define
$$ g(\mu_2,H,\sigma)= \left( \frac{2  \mu_2}{\sigma^2 \Gamma(2H+1)}  \right)^{-\frac{1}{2H}}.$$
The derivatives of $g$  with respect to $\sigma$, $H$ and $\mu_2$ are bounded when $0<\Lambda <+\infty$, $0< \underline{\sigma}<\overline{\sigma}< +\infty$ and $\frac{1}{2} < H < \frac{3}{4}$.  
Therefore, as $\Delta_N  \left( \log N \right)^2 \longrightarrow 0$  as $N \longrightarrow + \infty$, we can obtain by Taylor expansion that $$ \sqrt{T_N} \left( g(M_N) -g(m) \right) \stackrel{\cal L}{\longrightarrow} {\cal N}(0, g_{\mu_2}^\prime(m)^2   \left( \sigma_H \kappa_H \right)^2 \lambda^{-4H-1})$$
or 
$$ \sqrt{T_N} \left(\widehat{\lambda}_N-\lambda \right) \stackrel{\cal L}{\longrightarrow} {\cal N}(0,\Gamma_3(\vartheta))$$
where $\Gamma_3(\vartheta)= g^\prime_{\mu_2}(m)^2  \left( \sigma_H \kappa_H \right)^2 \lambda^{-4H-1} = \lambda \left( \frac{\sigma_H}{2H} \right)^2$, $g^\prime_{\mu_2}(.)$ is the derivative of $g$ with respect to $\mu_2$ and
$$ \widehat{\lambda}_N \stackrel{a.s.}{\longrightarrow}  \lambda$$
as $N \longrightarrow + \infty$.


\end{proof}

\begin{remark}
The different conditions on $\Delta_N$ raise the question of whether such a rate actually exists.
One possible mesh is $\Delta_N= \frac{\log N}{N}$.
\end{remark}

\begin{remark}\label{rem:sigma}
As in the classical case $H=\frac12$, the limit variance $\Gamma_3(\vartheta)$ does not depend on the diffusion coefficient $\sigma$. Let us also notice that the quantity $\sigma_H^2$ appearing in $\Gamma_3(\vartheta)$ is an increasing function of $H$.
\end{remark}

\section{Statistical software and Monte-Carlo analysis}\label{sec3}
In this section we present a brief introduction to the \texttt{yuima} package for \texttt{R} statistical environment \cite{R}. The   \texttt{yuima} package is a comprehensive framework, based on the S4 system of classes and methods, which allows for the description of solutions of stochastic differential equations. Although we cannot give details here, the user can specify a stochastic differential equation of the form
$$
dX_t = b(t,X_t)dt + \sigma(t,X_t)dW_t^H + c(t,X_t)Z_t
$$
where  the coefficients $b(\cdot,\cdot)$, $\sigma(\cdot, \cdot)$ and $c(\cdot,\cdot)$ are entirely specified by the user, even in parametric form; $(Z_t, t \geq 0)$ is a L\'evy process (for more information on L\'evy processes, see \cite{bertoin,sato} and $(W_t^H, t \geq 0)$ is a fractional Brownian motion (recall that $(W_t^\frac12, t \geq 0)$ is the standard Brownian motion). The L\'evy process $(Z_t, t \geq 0)$ and the fractional Brownian motion $(W_t^H, t \geq 0)$ can be present at the same time only when $H=\frac12$, but all other combinations are possible.
The  \texttt{yuima} package provides the user, not only the simulation part, but also several parametric and non-parametric estimation procedures. In the next section we present an example of use only for simulation and estimation of the fractional Ornstein-Uhlenbeck process considered in this paper.

To test the performance of the estimators for finite samples, we run a Monte-Carlo analysis. We consider different setup for the parameters even outside the region $\frac12 <H <\frac34$ and different sample size with large and small values of $T$ in order to test the performance of the estimator of the drift parameter when  the stationarity is not reached by the process.
All numerical experiments presented in the following  have been done with the \texttt{yuima} package \cite{YUIMA}.

\subsection{Example of numerical simulation and estimation of the fOU process with the \texttt{yuima} package}\label{sec:sim1}
With the  \texttt{yuima} package the fractional Gaussian noise is simulated with the Wood and Chan method \cite{fGnSim2} or other techniques.
We present below how to simulate one sample path of the fractional Ornstein-Uhlenbeck process with Euler-Maruyama method.
For instance, loading the package with
\begin{verbatim}
library(yuima)
\end{verbatim}
we can simulate a regularly sampled path of the following model
$$ X_t= 1 - 2 \int_0^t X_t dt +  dW^H_t, \qquad H=0.7, $$
with
\begin{verbatim}
samp <-setSampling(Terminal=100, n=10000)
mod <- setModel(drift="-2*x", diffusion="1",hurst=0.7)
ou <- setYuima(model=mod, sampling=samp)
fou <- simulate(ou, xinit=1)
\end{verbatim}
The estimation procedure of the Hurst parameter have been implemented in \texttt{qgv} function. 
 In order to estimate only the parameter $H$, one can use 
\begin{verbatim}
qgv(fou)
\end{verbatim}
that works also for non linear fractional diffusions (see \cite{phd}). 
The  procedure for joint estimation of
the Hurst exponent $H$, diffusion coefficient $\sigma$ and drift parameter $\lambda$ is called
\texttt{lse(,frac=TRUE)}.  So for example, in order to estimate the three different parameters $H$, $\lambda$ and $\sigma$, one can use 
\begin{verbatim}
lse(fou,frac=TRUE)
\end{verbatim}
which uses by default the order 2  Daubechies filter (see Remark \ref{rem:1}) if the user does not specify the \texttt{filter} argument.

\subsection{Performance of the Hurst parameter and diffusion coefficient estimation}

In this first simulation part, we present mean average values and standard deviation values for both estimators $\widehat{H}_N$ and $\widehat{\sigma}_N$ (see Section \ref{sec:H} for the definitions) 
with 500 Monte-Carlo replications. This have been done for different Hurst exponents $H$ and different
diffusion coefficients $\sigma$ in the model \eqref{eq:model}, the parameter $\lambda$ being fixed equal to 2. The results are presented in Table \ref{tab:1} and Table \ref{tab:2} for different values of the horizon time $T_N$ and the sample size $N$.

\begin{table}[ht] \small
\centering
\begin{tabular}{|c|c|c|c|}
\hline
$\widehat{H}$\phantom{$\widehat{H}$} & $H=0.5$ & $H=0.7$ & $H=0.9$\\
\hline
\hline
$\sigma=1$  & 0.499 & 0.697   &  0.898  \\
 & (0.035) & (0.033)  & (0.031) \\ 
\hline
$\sigma=2$&0.498 & 0.700 & 0.898\\
& (0.033)  &  (0.034) & (0.033)\\
\hline
\end{tabular}
\hfill
\begin{tabular}{|c|c|c|c|}
\hline
$\widehat{\sigma}$\phantom{$\widehat{H}$} & $H=0.5$ & $H=0.7$ & $H=0.9$\\
\hline
\hline
 
$\sigma=1$ &    1.024 & 1.016 & 1.081\\
& (0.262) & (0.282) & (0.437) \\
\hline

\hline
$\sigma=2$& 2.035 & 2.073 & 2.213 \\
& (0.510) & (0.564) & (1.110)\\
\hline
\end{tabular}
\caption{Mean average (and standard deviation in parenthesis) of 500 Monte-Carlo simulations
for the estimation of $H$ (left) and $\sigma$ (right) for different cases. Here $T=100$, $N=1000$ and $\lambda=2$.}\label{tab:1}
\end{table}

\begin{table}[ht]\small
\centering
\begin{tabular}{|c|c|c|c|}
\hline
$\widehat{H}$\phantom{$\widehat{H}$} & $H=0.5$ & $H=0.7$ & $H=0.9$\\
\hline
\hline
$\sigma=1$  & 0.500 & 0.700   &  0.900  \\
 & (0.003) & (0.003)  & (0.003) \\ 
\hline
$\sigma=2$&0.500 & 0.700 & 0.900\\
& (0.004)  &  (0.003) & (0.003)\\
\hline
\end{tabular}
\hfill
\begin{tabular}{|c|c|c|c|}
\hline
$\widehat{\sigma}$\phantom{$\widehat{H}$} & $H=0.5$ & $H=0.7$ & $H=0.9$\\
\hline
\hline
 
$\sigma=1$ &    1.000 & 1.001 & 0.999\\
& (0.025) & (0.026) & (0.036) \\
\hline
$\sigma=2$& 2.001 & 2.002 & 1.997 \\
& (0.053) & (0.053) & (0.073)\\
\hline
\end{tabular}
\caption{Mean average (and standard deviation in parenthesis) of 500 Monte-Carlo simulations
for the estimation of $H$ (left) and $\sigma$ (right) for different cases, and for $T_N=100$, $N=100000$ and $\lambda=2$.}\label{tab:2}
\end{table}

Contrary to the estimation of the drift (see Section \ref{sim:lambda}), we have consistent estimates of $H$ and $\sigma$ for any values of $T_N$. Only the size of the sample $N$ have influence on the performance of the estimate.

\subsection{Plug-in for the estimation of  drift parameter $\lambda$}\label{sim:lambda}

In this second simulation part, we present  mean average values and standard deviation values for the estimator $\widehat{\lambda}_N$ (see Section \ref{sec:lambda} for the definition) of  the drift with 500 Monte-Carlo replications. This have been done for different values of $\lambda$ and $H$ in model \eqref{eq:model}, the diffusion coefficient $\sigma$ being fixed to 1 (see Remark \ref{rem:sigma}). The results are presented in Table \ref{tab:3} for different values of the horizon time $T_N$ and the sample size $N$.

\begin{table}[ht] \small
\centering
\begin{tabular}{|c|c|c|c|}
\hline
& $H=0.5$ & $H=0.6$ & $H=0.7$\\
\hline
\hline
$\lambda=0.5$  & 0.093 &  0.214 &  0.353 \\
& (0.037) & (0.057) &  (0.069) \\
\hline
$\lambda=1$& 0.138 &  0.276 &  0.432 \\
& (0.052) &  (0.068) & (0.078) \\
\hline
\end{tabular}
\hfill
\begin{tabular}{|c|c|c|c|}
\hline
& $H=0.5$ & $H=0.6$ & $H=0.7$\\
\hline
\hline
$\lambda=0.5$ & 0.476 & 0.514 & 0.605\\
& (0.148) & (0.166) & (0.298)\\
\hline
$\lambda=1$& 0.906 &  0.940 & 1.005 \\
& (0.227) & (0.238) & (0.412)\\
\hline
\end{tabular}
\caption{Mean average (and standard deviation in parenthesis) of 500 Monte-Carlo simulation
for the estimation of $\lambda$ for different values of $H$ and $\lambda$.  
Here $\sigma=1$ and
$T_N=1$ and 
$N=100000$ (left) and $T_N=100$ and 
$N=1000$ (right).}\label{tab:3}
\end{table}

We can see in Table \ref{tab:3} that the values of $T_N$ is important for the estimation of the drift. Actually, the consistency of the estimates are valid for increasing values of $T_N$ and decreasing values of the mesh size $\Delta_N$. Moreover, the bigger $H$, the harder  the estimation of the drift parameter. This phenomena can be explained by the long-range dependence property of the fOU process. It is the same for $\lambda$ ; as $\lambda$ increases, its estimation is harder (see Remark \ref{rem:sigma}).  It can be explained by the fact that when $\lambda$ is bigger, the fOU process enters faster in its stationary behavior where it is more difficult to detect the trend.

Finally, in order to illustrate the asymptotical normality for the estimator $\widehat{\lambda}$ of $\lambda$, we present in Figure \ref{fig:density} the kernel estimation of the density.

\begin{figure}[ht]
\centering
\includegraphics[width=10cm]{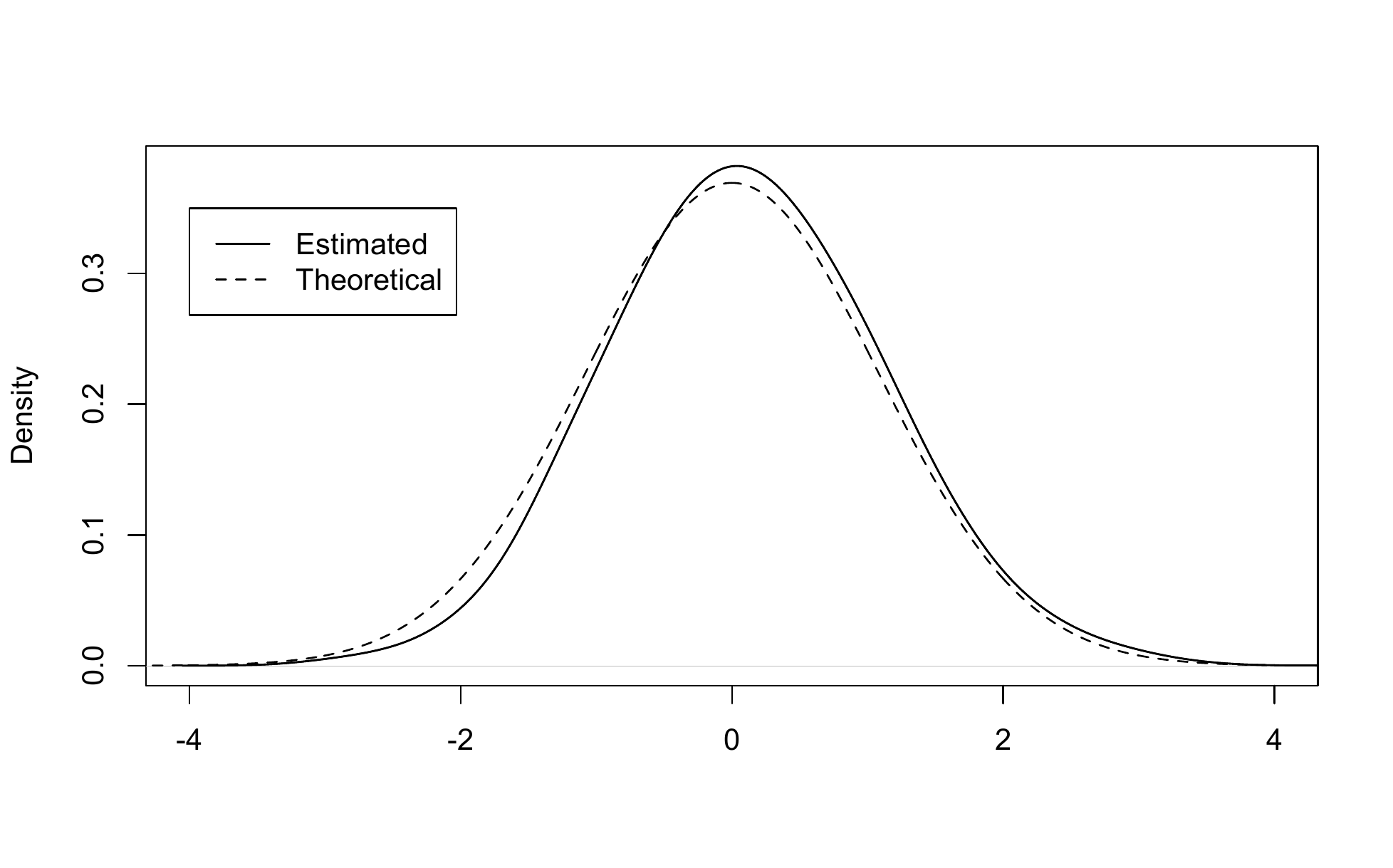}
\caption{Kernel estimation for the density of $\left( \sqrt{T_N} \left( \widehat{\lambda}^{(m)}_N-\lambda \right)\right)_{m=1\ldots M}$, $M=5000$, for $T_N=1000$ and $T_N=100000$ (fill line) and the theoretical Gaussian density ${\cal N}(0,\Gamma_3(\vartheta))$ (dashed line) for $\vartheta=(\lambda,\sigma,H)=(0.3,1,0.7)$ (for the value of  $\Gamma_3(\vartheta)$ see Theorem \ref{thm:lambda}).}\label{fig:density}
\end{figure}

\section*{Acknowledgments} 
We would like to thank Marina Kleptsyna for the discussions and her interest for this work. Computing resources have been financed by Mostapad project in
CNRS FR 2962. This work has been supported by the project PRIN 2009JW2STY, Ministero dell'Istruzione dell'Universit\`a e della Ricerca.

\bibliographystyle{plain}
\bibliography{recherchebib}

\begin{thebibliography}{10}

\bibitem{Coutin2}
B.~Bercu, L.~Coutin, and N.~Savy.
\newblock Sharp large deviations for the fractional {O}rnstein-{U}hlenbeck
  process.
\newblock {\em Teoriya Veroyatnostei i ee Primeneniya}, 2010.

\bibitem{bertoin}
J.~Bertoin.
\newblock {\em L\'evy Processes}.
\newblock Cambridge University Press, Cambridge, 1998.

\bibitem{berzinleon}
C.~Berzin and J.~Leon.
\newblock Estimation in models driven by fractional brownian motion.
\newblock {\em Annales de l'Institut Henri Poincar\'e}, 44(2):191--213, 2008.

\bibitem{BK10}
A.~Brouste and M.~Kleptsyna.
\newblock Asymptotic properties of {MLE} for partially observed fractional
  diffusion system.
\newblock {\em Statistical {I}nference for {S}tochastic {P}rocesses},
  13(1):1--13, 2010.

\bibitem{cheridito}
P.~Cheridito, H.~Kawaguchi, and M.~Maejima.
\newblock Fractional {O}rnstein-{U}hlenbeck processes.
\newblock {\em Electronic Journal of Probability}, 8(3):1--14, 2003.

\bibitem{Lototsky2}
I.~Cialenco, S.~Lototsky, and J.~Pospisil.
\newblock Asymptotic properties of the maximum likelihood estimator for
  stochastic parabolic equations with additive fractional {B}rownian motion.
\newblock {\em Stochastics and Dynamics}, 9(2):169--185, 2009.

\bibitem{daub}
I.~Daubechies.
\newblock {\em Ten Lectures on Wavelets}.
\newblock SIAM, 1992.

\bibitem{FZ}
D.~Florens-Zmirou.
\newblock Approximate discrete time schemes for statistics of diffusion
  processes.
\newblock {\em Statistics}, 20:263--284, 1989.

\bibitem{GC}
V.~Genon-Catalot.
\newblock Maximum constrast estimation for diffusion processes from discrete
  observation.
\newblock {\em Statistics}, 21:99--116, 1990.

\bibitem{gobet1}
E.~Gobet.
\newblock Lan property for ergodic diffusions with discrete observations.
\newblock {\em Annales de l'Institut Henri Poincar\'e}, 38(5):711--737, 2002.

\bibitem{hunual}
Y.~Hu and D.~Nualart.
\newblock Parameter estimation for fractional ornstein-uhlenbeck processes.
\newblock {\em Statistics and Probability Letters}, 80(11-12):1030--1038, 2010.

\bibitem{istaslang}
J.~Istas and G.~Lang.
\newblock Quadratic variations and estimation of the local hšlder index of a
  gaussian process.
\newblock {\em Annales de l'Institut Henri Poincar\'e}, 23(4):407--436, 1997.

\bibitem{Jacod1}
J.~Jacod.
\newblock Inference for stochastic processes.
\newblock {\em Statistics}, Prepublication 683, 2001.

\bibitem{Kess}
M.~Kessler.
\newblock Estimation of an ergodic diffusion from discrete observations.
\newblock {\em Scandinavian Journal of Statistics}, 24:211--229, 1997.

\bibitem{Klep02a}
M.~Kleptsyna and A.~Le Breton.
\newblock Statistical analysis of the fractional {O}rnstein-{U}hlenbeck type
  process.
\newblock {\em Statistical Inference for Stochastic Processes}, 5:229--241,
  2002.

\bibitem{kolmogorov}
A.~Kolmogorov.
\newblock Winersche {S}piralen und einige andere interessante {K}urven in
  {H}ilbertschen {R}aum.
\newblock {\em Acad. Sci. USSR}, 26:115--118, 1940.

\bibitem{ludena}
C.~Ludena.
\newblock Minimum contrast estimation for fractional diffusion.
\newblock {\em Scandinavian Journal of Statistics}, 31:613--628, 2004.

\bibitem{mandel}
B.~Mandelbrot and J.~Van Ness.
\newblock Fractional {B}rownian motions, fractional noises and application.
\newblock {\em SIAM Review}, 10:422--437, 1968.

\bibitem{phd}
D.~Melichov.
\newblock {\em On estimation of the {H}urst index of solutions of stochastic
  equations}.
\newblock PhD thesis, Vilnius Gediminas Technical University, 2011.

\bibitem{Neue}
A.~Neuenkirch and S.~Tindel.
\newblock A {L}east {S}quare-type procedure for parameter estimation in
  stochastic differential equations with additive fractional noise.
\newblock {\em preprint}, 2011.

\bibitem{R}
{R Development Core Team}.
\newblock {\em R: A Language and Environment for Statistical Computing}.
\newblock R Foundation for Statistical Computing, Vienna, Austria, 2010.
\newblock {ISBN} 3-900051-07-0.

\bibitem{sato}
K.~Sato.
\newblock {\em L\'evy Processes and Infinitely Divisible Distributions}.
\newblock Cambridge University Press, Cambridge, 1999.

\bibitem{YUIMA}
YUIMA~Project Team.
\newblock {\em yuima: The YUIMA Project package (unstable version)}, 2011.
\newblock R package version 0.1.1936.

\bibitem{fGnSim2}
A.~Wood and G.~Chan.
\newblock Simulation of stationary {G}aussian processes.
\newblock {\em Journal of computational and graphical statistics},
  3(4):409--432, 1994.

\bibitem{chinois}
W.~Xiao, W.~Zhang, and W.~Xu.
\newblock Parameter estimation for fractional ornsteinÐuhlenbeck processes at
  discrete observation.
\newblock {\em Applied Mathematical Modelling}, 35:4196--4207, 2011.

\bibitem{Yoshida1}
N.~Yoshida.
\newblock Estimation for diffusion processes from discrete observations.
\newblock {\em Journal of Multivariate Analysis}, 41:220--242, 1992.

\end{thebibliography}

\end{document}